\documentclass[11pt]{article}
\usepackage{pdfsync}
\usepackage{amsmath}
\usepackage{amssymb}
\usepackage{amscd}
\usepackage{mathrsfs}
\usepackage{color}
\usepackage{ulem}
\usepackage{bm}
\usepackage{graphicx}
\usepackage{longtable}
\newtheorem{thm}{Theorem}
\newtheorem{lem}[thm]{Lemma}
\newtheorem{cor}[thm]{Corollary}
\newtheorem{prop}[thm]{Proposition}
\newtheorem{definition}[thm]{Definition}
\newtheorem{remark}[thm]{Remark}

\newtheorem{example}[thm]{Example}
\newtheorem{algorithm}[thm]{Algorithm}
\newtheorem{problem}[thm]{Problem}
\newtheorem{notation}[thm]{Notation}
\def\QED{\hbox{\hskip 1pt \vrule width4pt height 6pt depth 1.5pt \hskip 1pt}}
\newenvironment{prf}[1]{\trivlist
\item[\hskip \labelsep{\bf #1.\hspace*{.3em}}]}{~\hspace{\fill}~$\square$\endtrivlist}
\newenvironment{proof}{
\begin{prf}{Proof}}{
\end{prf}}
 \def\square{\QED}

\def\bZ{{\mathbb Z}}

\def\ord{\mbox{ord }}
\def\GL{{\rm GL}}

\def\lclm{{\rm lclm}}
\def\gcrd{{\rm gcrd}}

\def\fraka{{\mathfrak a}}

\def\frakm{{\mathfrak m}}

\def\frakD{{\mathfrak D}}

\def\frakW{{\mathfrak W}}

\def\QED{\hbox{\hskip 1pt \vrule width4pt height 6pt depth 1.5pt \hskip 1pt}}

\def\calU{{\cal U}}
\def\calM{{\cal M}}

\def\ann{{\rm ann}}

\def\va{{\mathbf a}}

\def\vc{{\mathbf c}}

\def\vg{{\mathbf g}}
\def\vp{{\mathbf p}}
\def\vh{{\mathbf h}}
\def\vd{{\mathbf d}}
\def\ve{{\mathbf e}}

\def\vz{{\mathbf z}}
\def\vv{{\mathbf v}}

\def\pa{{\partial}}
\def\vpa{{\bm \partial}}

\def\ord{{\rm ord}}
\def\rd{{\rm d}}

\def\bfx{\mathbf x}
\def\bfc{\mathbf c}

\newcommand{\vx} {{\bf x}}
\newcommand{\Mat} {{\rm Mat}}

\begin{document}
\title{Telescopers for differential forms with one parameter\thanks{S.\ Chen was partially supported by the NSFC
grants 11871067, 11688101, the Fund of the Youth Innovation Promotion Association, CAS, and the National Key Research and Development Project 2020YFA0712300, R. Feng was partially supported by the NSFC
grants  11771433, 11688101, Beijing Natural Science Foundation under Grant Z190004, and the National Key Research and Development Project 2020YFA0712300.}}
\author{\smallskip Shaoshi Chen$^{1}$, \,  Ruyong Feng$^{1}$, \,  Ziming Li$^1$, \\  \bigskip
Michael F.\ Singer$^{2}$, \, Stephen Watt$^3$ \\
$^1$KLMM,\, AMSS, \,Chinese Academy of Sciences and\\ School of Mathematics, University of Chinese Academy of Sciences, \\Beijing, 100190, China\\
$^2$Department of Mathematics, North Carolina State University,\\ Raleigh, 27695-8205, USA\medskip \\
$^3$SCG, Faculty of Mathematics,  University of Waterloo, \\Ontario, N2L3G1, Canada\\
{\sf \{schen, ryfeng\}@amss.ac.cn, zmli@mmrc.iss.ac.cn}\\
{\sf  singer@math.ncsu.edu, smwatt@uwaterloo.ca}
%Preliminary notes
}

\maketitle
\begin{abstract}
Telescopers for a function are linear differential (resp. difference) operators annihilated by the definite integral (resp. definite sum) of this function. They play a key role in Wilf-Zeilberger theory and algorithms for computing them have been extensively studied in the past thirty years. In this paper, we introduce the notion of telescopers for differential forms with $D$-finite function coefficients. These telescopers appear in several areas of mathematics, for instance parametrized differential Galois theory and mirror symmetry. We give a sufficient and necessary condition for the existence of telescopers for a differential form and describe a method to compute them if they exist. Algorithms for verifying this condition are also given.
\end{abstract}
\section{Introduction} \label{SECT:intro}
In the Wilf-Zeilberger theory, telescopers usually refer to the operators in the output of the method of creative telescoping, which are linear differential (resp. difference) operators annihilated by the definite integrals (resp. the definite sums) of the input functions. The telescopers have emerged at least from the work of Euler \cite{euler} and have been found many applications in the various areas of mathematics such as combinatorics, number theory, knot theory and so on (see Section 7 of \cite{koutschan} for details). In particular, telescopers for a function are often used to prove the identities involving this function or even obtain a simpler expression for the definite integral or sum of this function. As a clever and algorithmic process for constructing telescopers, creative telescoping firstly appeared as a term in the essay of van der Poorten on Apr\'ey's proof of the irrationality of $\zeta(3)$ \cite{vanderpoorten}. However, it was Zeilberger and his collaborators \cite{almkvist-zeilberger,petkovsek-wilf-zeilberger,wilf-zeilberger1,wilf-zeilberger2,zeilberger} in the early 1990s who equipped creative telescoping with a concrete meaning and formulated it as an algorithmic tool. Since then, algorithms for creative telescoping have been extensively studied. Based on the techniques used in the algorithms, the existing algorithms are divided into four generations, see \cite{chen-kauers} for the details. Most recent algorithms are called reduction-based algorithms which were first introduced by Bostan et al. in \cite{bostan-chen-chyzak} and further developed in \cite{bostan-chen-chyzak-li-xin,chen-kauers-koutschan,chen-vanhoeij-kauers-koutschan,bostan-chyzak-lairez-salvy} etc. The termination of these algorithms relies on the existence of telescopers. The question for which input functions the algorithms will terminate has been answered in \cite{wilf-zeilberger3,abramov2,abramov-le,chen-hou-mu,chen-chyzak-feng-fu-li} etc for several classes of functions such as rational functions and hypergeometric functions and so on. The algorithmic framework for creative telescoping is now called the Wilf-Zeilberger theory.

Most of algorithms for creative telescoping focus on the case of one bivariate function as input. There are only a few algorithms which deal with multivariate case (see \cite{chen-feng-li-singer,bostan-lairez-salvy,lairez,chen-hou-labahn-wang} etc). It is still a challenge to develop the multivariate analogue of the existing algorithms (see Section 5 of \cite{chen-kauers}). In the language of differential forms (with $m$ variables and one parameter), the results in \cite{chen-feng-li-singer} and \cite{lairez} dealt with the cases of differential 1-forms and differential $m$-forms respectively. On the other hand, in the applications to other domains such as mirror symmetry (see \cite{li-lian-yau,morrison-walcher,muller-weinzierl-zayadeh}), one needs to deal with the case of differential $p$-forms with $1\leq p \leq m$. Below is an example.
%\begin{example}
%Consider
%\begin{equation}
%\label{EQ:orderoneeqns}
%   \frac{\partial Y}{\partial x_1}=f_1,\cdots, \frac{\partial Y}{\partial x_n}=f_n
%\end{equation}
%where the $f_i\in k(t, x_1,\cdots,x_n)$ satisfy the compatibility condition, i.e. for all $1\leq i <j\leq n$,
%$\partial f_j/\partial x_i=\partial f_i/\partial x_j.$ Write $\omega=f_1\rd x_1+\cdots+f_n \rd x_n$.
% It was shown in \cite{} that the parameterized Galois group of (\ref{EQ:orderoneeqns}) over the differential closure of $k(t)$ is determined by a nonzero linear differential operator $L$ in $\partial /\partial t$ with polynomial coefficients in $t$ satisfying that
%$$
%   L(\omega)=\rd g
%$$
%where $g\in k(t,x_1,\cdots,x_n)$. Such $L$ is called a parallel telescoper for $f_1,\cdots,f_n$ with respect to $x_1,\cdots,x_n$ in \cite{}.
%\end{example}
\begin{example}
\label{EX:calabi-yau}
Consider the following one-parameter family of the quintic polynomials
$$
    W(t)=\frac{1}{5}(x_1^5+x_2^5+x_3^5+x_4^5+x_5^5) -t x_1x_2x_3x_4x_5
$$
where $t$ is a parameter. Set
$$
   \omega=\sum_{i=1}^5 \frac{(-1)^{i-1} x_i}{W(t)} \rd x_1\wedge \cdots \wedge \widehat{\rd x_i} \wedge \cdots \wedge \rd x_5.
$$
To obtain the Picard-Fuchs equation for the mirror quintic, the geometriests want to compute a fourth order linear differential operator $L$ in $t$ and $\pa_t$ such that $L(\omega)=\rd \eta$ for some differential 3-form $\eta$. Here one has that
$$
   L=(1-t^5)\frac{\partial^4}{\partial t^4}-10t^4\frac{\partial^3}{\partial t^3}-25t^3\frac{\partial^2}{\partial t^2}-15t^2\frac{\partial}{\partial t}-1.
$$
Set $\theta_t=t\partial/\partial t$. Then
$$
   \tilde{L}=-\frac{1}{5^4}L\frac{1}{t}=\theta_t^4-5t(5\theta_t+1)(5\theta_t+2)(5\theta_t+3)(5\theta_t+4)
$$
and the equation $\tilde{L}(y)=0$ is the required Picard-Fuchs equation.
\end{example}
We call the operator $L$ appearing in the above example a telescoper for the differential form $\omega$ (see Definition~\ref{DEF:telescopers}). In this paper, we study the telescopers for differential forms with $D$-finite function coefficients. Instead of the geometric method used in \cite{li-lian-yau,morrison-walcher,muller-weinzierl-zayadeh}, we provide an algebraic treatment. We give a sufficient and necessary condition guaranteeing the existence of telescopers and describe a method to compute them if they exist. Meanwhile, we also present algorithms to verify this condition.

The rest of this paper is organized as follows. In Section 2, we recall differential forms with $D$-finite function coefficients and introduce the notion of telescopers for differential forms. In Section 3, we give a sufficient and necessary condition for the existence of telescopers, which can be considered as a parametrized version of Poincar\'{e} lemma on differential manifolds. In Section 4, we give two algorithms for verifying the condition presented in Section 3.

{\bf Notations}: The following notations will be frequently used thoughout this paper.
\begin{longtable}{rl}
$\pa_t$: &the usual derivation $\pa/\pa_t$ with respect to $t$,\\
$\pa_{x_i}$:& the usual derivation $\pa/\pa_{x_i}$with respect to $x_i$,\\
$\vx$:& $\{x_1,\cdots,x_n\}$\\
$\vpa_\vx$:& $\{\pa_{x_1},\cdots,\pa_{x_n}\}$,\\
\end{longtable}
The following formulas will also be frequently used:
\begin{align}
   \pa_x^{\mu} x^{\nu} &= \begin{cases} \nu(\nu-1)\cdots(\nu-\mu+1)x^{\nu-\mu} + *\pa_x, & \nu \geq \mu \\
     * \pa_x, & \nu<\mu  \end{cases}  \label{EQ:formula1} \\
      x^{\mu}\pa_x^{\nu}  &=\begin{cases}  (-1)^{\nu}\mu(\mu-1)\cdots(\mu-\nu+1)x^{\mu-\nu} +\pa_x *, & \mu\geq \nu\\
     \pa_x *,  &\mu<\nu \end{cases} \label{EQ:formula2}
\end{align}
where $*\in k\langle x,\pa_x\rangle$.

\section{$D$-finite elements and differential forms}
Throughout this paper, let $k$ be an algebraically closed field of characteristic zero and let $K$ be the differential field $k(t,x_1,\cdots,x_n)$ with the derivations $\pa_t, \pa_{x_1}, \cdots,\pa_{x_n}$. Let $\frakD=K\langle \pa_t, \vpa_\vx\rangle$ be the ring of linear differential operators with coefficients in $K$. For $S\subset \{t, \vx, \pa_t, \vpa_\vx\}$, denote by $k\langle S \rangle$ the subalgebra over $k$ of $\frakD$ generated by $S$. For brevity, we denote $k\langle t, \vx, \pa_t, \pa_\vx\rangle$  by $\frakW$. Let $\calU$ be the universal differential extension of $K$ in which every algebraic differential equation having a solution in an extension of $\calU$ has a solution (see page 133 of \cite{kolchin} for more precise description).
\begin{definition}
An element $f\in \calU$ is said to be $D$-finite over $K$ if for every $\delta\in \{\pa_t, \pa_{x_1}, \cdots, \pa_{x_n}\}$, there is a nonzero operator $L_{\delta}\in K\langle \delta \rangle$ such that $L_{\delta}(f)=0$.
\end{definition}

Denote by $R$ the ring of $D$-finite elements over $K$, and by $\calM$ a free $R$-module of rank $m$ with base $\{\fraka_1,\cdots,\fraka_m\}$. Define a map $\frakD\times \calM\rightarrow \calM$ given by
$$\left(L,\sum_{i=1}^mf_i\fraka_i\right)\rightarrow L\left(\sum_{i=1}^mf_i\fraka_i\right):=\sum_{i=1}^mL(f_i)\fraka_i.$$
This map endows $\calM$ with a left $\frakD$-module structure.
Let
\[
  \bigwedge(\calM)=\bigoplus_{i=0}^m \bigwedge \nolimits^i(\calM)
\]
 be the exterior algebra of $\calM$, where $\bigwedge^i(\calM)$ denotes the $i$-th homogeneous part of $\bigwedge(\calM)$ as a graded $R$-algebra. We call an element in $\bigwedge^i(\calM)$ an $i$-form. $\bigwedge(\calM)$ is also a left $\frakD$-module.
Let $\rd: R\rightarrow \calM$ be a map defined as
\[
   \rd f=\pa_{x_1}(f)\fraka_1+\cdots+\pa_{x_m}(f)\fraka_m
\]
for any $f\in R$. Then $\rd$ is a derivation over $k$. Note that for each $i=1,\cdots,m$, $\rd x_i=\fraka_i$. Hence in the rest of this paper we shall use $\{\rd x_1,\cdots,\rd x_m\}$ instead of $\{\fraka_1,\cdots,\fraka_m\}$. The map $\rd$ can be extended to a derivation on $\bigwedge(\calM)$ which is defined recursively as
$$
   \rd(\omega_1\wedge\omega_2)=\rd\omega_1\wedge\omega_2+(-1)^{i-1}\omega_1\wedge \rd\omega_2
$$
for any $\omega_1\in \bigwedge^i(\calM)$ and $\omega_2\in \bigwedge^j(\calM)$. For detailed definitions on exterior algebra and differential forms, we refer the readers to Chapter 19 of \cite{lang} and  Chapter 1 of \cite{weinstraub} respectively.
As the usual differential forms, we introduce the following definition.
\begin{definition} Let $\omega\in \bigwedge(\calM)$ be a form.
\begin{itemize}
\item [$(1)$]
$\omega$ is said to be closed if $\rd\omega=0$, and exact if there is  $\eta\in \bigwedge(\calM)$ such that $\omega=\rd\eta$.
\item [$(2)$]
$\omega$ is said to be $\pa_t$-closed ($\pa_t$-exact) if there is a nonzero $L\in k(t)\langle \pa_t\rangle$ such that $L(\omega)$ is closed (exact).
\end{itemize}
\end{definition}
\begin{definition}
\label{DEF:telescopers}
Assume that $\omega\in \bigwedge(\calM)$. A nonzero $L\in k(t)\langle \pa_t\rangle$ is called a telescoper for $\omega$ if $L(\omega)$ is exact.
\end{definition}
\section{Parametrized Poincar\'{e} lemma}
The famous Poincar\'e lemma states that if $B$ is an open ball in $\mathbb{R}^n$, any smooth closed $i$-form $\omega$ defined on $B$ is exact, for any integer $i$ with $1\leq i \leq n$. In this section, we shall prove the following lemma which can be viewed as a parametrized analogue of Poincar\'e lemma for $\bigwedge(\calM)$.

\begin{lem}[Parameterized Poincar\'{e} lemma]
\label{LM:ppl}
Let~$\omega  \in \bigwedge^p(\calM)$. If $\omega$ is $\pa_t$-closed then it is $\pa_t$-exact.
\end{lem}
To the above lemma, we need some lemmas.
\begin{lem}[Lipshitz's lemma (Lemma 3 of \cite{lipshitz})]
Assume that $f$ is a $D$-finite element over $k(\vx)$. For each pair $1\leq i <j \leq n$, there is a nonzero operator $L\in k(x_1,x_3,\cdots, x_n)\langle \pa_{x_i}, \pa_{x_j}\rangle$ such that $L(f)=0$.
\end{lem}
The following lemma is a generalization of Lipshitz's lemma.
\begin{lem}
\label{LM:modifiedlipshitz}
Assume that $f_1,\cdots, f_m$  are $D$-finite elements over $k(\vx,t)$ and
$$S\subset \{t,x_1,\cdots,x_n,\pa_t,\pa_{x_1},\cdots,\pa_{x_n}\}$$
with $|S|>n+1$. Then one can compute a nonzero operator $T$ in $k\langle S\rangle$ such that $T(f_i)=0$ for all $i=1,\cdots,m$.
\end{lem}
\begin{proof}
For each $\delta\in \{t,\pa_{x_1},\cdots,\pa_{x_n}\}$ and $i=1,\cdots,m$, let $T_i$ be a nonzero operator in $K\langle \delta \rangle$ such that $T_i(f_i)=0$. Set $T$ to be the least common left multiple of $T_1,\dots,T_m$. Then $T(f_i)=0$ for all $i=1,\cdots,m$. The lemma then follows from an argument similar to that in the proof of Lipshitz's lemma.
\end{proof}
\begin{lem}
\label{LM:basecase}
Assume that $f_1,\cdots,f_m$ are $D$-finite over $k(\vx,t)$, $I,J\subset \{1,\cdots,n\}$ and $I\cap J=\emptyset$. Assume further that $V\subset \{x_i,\pa_{x_i} | i\in \{1,\cdots,n\}\setminus (I\cup J)\}$ with $|V|=n-|I|-|J|$. Then one can compute an operator $P$ of the form
\[
    L+\sum_{i\in I} \pa_{x_i} M_i +\sum_{j\in J} N_j  \pa_{x_j}
\]
such that $P(f_l)=0$ for all $l=1,\cdots,m$, where $L$ is a nonzero operator in $k\langle \{t, \pa_t\}\cup V\}\rangle$, $M_i,N_j\in \frakW$
   and $N_j$ is free of $x_i$ for all $i\in I$ and $j\in J$.
\end{lem}
\begin{proof}
Without loss of generality, we assume that $I=\{1,\cdots,r\}$ and $J=\{r+1,\cdots,r+s\}$ where $r=|I|$ and $s=|J|$.
Let $$S=\{t,\pa_t\}\cup \{\pa_{x_i} | i\in I\}\cup \{x_j | j=r+1,\cdots,r+s\}\cup V.$$
Then $|S|=n+2>n+1$. By Lemma~\ref{LM:modifiedlipshitz}, one can compute a $T\in k\langle S\rangle \setminus \{0\}$ such that $T(f_l)=0$ for all $l=1,\cdots,m$. Write
$$
   T=\sum_{\vd=(d_1,\cdots,d_r)\in \Gamma_1}  \pa_{x_1}^{d_1}\cdots \pa_{x_r}^{d_r} T_{\vd}
$$
where $T_\vd \in k\langle \{t,\pa_t, x_{r+1},\cdots,x_{r+s}\}\cup V\}\rangle \setminus\{0\}$ and $\Gamma_1$ is a finite subset of $\bZ^r$. Let $\bar{\vd}=(\bar{d}_1,\cdots,\bar{d}_r)$ be the minimal element of $\Gamma_1$ with respect to the lex order on $\bZ^r$. Multiplying $T$ by $\prod_{i=1}^r x_i^{\bar{d}_i}$ on the left and using the formula (\ref{EQ:formula2}) yield that
\begin{equation}
\label{EQ:part1}
   \left(\prod_{i=1}^r x_i^{\bar{d}_i}\right)T=\alpha T_{\bar{\vd}} +\sum_{i=1}^{r}  \pa_{x_i} \tilde{T}_i
\end{equation}
where $\alpha$ is a nonzero integer and $\tilde{T}_i\in k\langle S\cup\{x_i |i\in I\} \rangle$. Write
$$
   T_{\bar{\vd}}=\sum_{\ve=(e_1,\cdots,e_s)\in \Gamma_2} L_{\ve} x_{r+1}^{e_1}\cdots x_{r+s}^{e_s}
$$
where $L_{\ve}\in k\langle \{t,\pa_t\}\cup V \rangle\setminus\{0\}$ and $\Gamma_2$ is a finite subset of $\bZ^s$. Let $\bar{\ve}=(\bar{e}_1,\cdots,\bar{e}_s)$ be the maximal element of $\Gamma_2$ with respect to the lex order on $\bZ^s$. Multiplying $T_{\bar{\vd}}$ by $\prod_{i=1}^s \pa_{x_{r+i}}^{\bar{e}_i}$ on the left and using the formula (\ref{EQ:formula1}) yield that
\begin{equation}
\label{EQ:part2}
    \left(\prod_{i=1}^s \pa_{x_{r+i}}^{\bar{e}_i}\right)T_{\bar{\vd}}=\beta L_{\bar{\ve}}+\sum_{j\in J} \tilde{L}_j \pa_{x_j}
\end{equation}
where $\tilde{L}_i\in k\langle \{t,\pa_t,x_{r+1},\cdots,x_{r+s},\pa_{x_{r+1}},\cdots,\pa_{x_{r+s}}\}\cup V\rangle$ and $\alpha$ is a nonzero integer. Combining (\ref{EQ:part1}) with (\ref{EQ:part2}) yields the required operator $P$.
\end{proof}
\begin{cor}
\label{COR:compatible}
Assume that $f_1,\cdots,f_m$ are $D$-finite over $k(\vx,t)$, $J$ is a subset of $\{1,\cdots,n\}$ and $V\subset \{x_i,\pa_{x_i} | i\in \{1,\cdots,n\}\setminus J\}$ with $|V|=n-|J|$. Assume further that $\pa_{x_j}(f_l)=0$ for all $j\in J$ and $l=1,\cdots,m$. Then one can compute a nonzero $L\in k\langle \{t,\pa_t\}\cup V \rangle$ such that $L(f_l)=0$ for all $l=1,\cdots,m$.
\end{cor}
\begin{proof}
In Lemma~\ref{LM:basecase}, set $I=\emptyset$.
\end{proof}
The main result of this section is the following theorem which can be viewed as a generalization of Corollary~\ref{COR:compatible} to differential forms. To describe and prove this theorem, let us recall some notation from the first chapter of \cite{weinstraub}. For any $f\in R$, we define $\rd_0(f)=0$ and
\[
   \rd_s(f) = \partial_{x_1}(f)\rd x_1 + \cdots +  \partial_{x_s}(f)\rd x_s
\]
for~$s\in \{1, 2, \ldots, n\}$.
We can extend~$\rd_s$ to the module $\bigwedge(\calM)$ in a natural way. Precisely, let $\omega=\sum_{i=1}^m f_i \frakm_i$ where $\frakm_i$ is a monomial in $\rd x_1,\cdots, \rd x_n$. Then $\rd_0(\omega)=0$ and
$$
   \rd_s(\omega)=\sum_{i=1}^m \sum_{j=1}^s \pa_{x_j}(f_i)\rd x_j\wedge \frakm_i=\sum_{j=1}^s \rd x_j \wedge \pa_{x_j}(\omega).
$$
By definition, one sees that
$$\rd_s(u\wedge \rd x_s)=\rd_{s-1}(u)\wedge \rd x_s\,\,\mbox{and}\,\,\rd_s(u)=\rd_{s-1}(u)+\rd x_s\wedge \pa_{x_s}(u).$$
\begin{thm}\label{THM:ppl}
Assume that~$0\leq s \leq n, V\subset \{x_{s+1},\cdots,x_n, \pa_{x_{x+1}},\cdots,\pa_{x_n}\}$ with $|V|=n-s$ and $\omega \in \bigwedge^p(\calM)$. If~$\rd_s \omega =0$, then one can compute a nonzero $L\in k\langle \{t, \partial_t\}\cup V\rangle$ and~$\mu \in \bigwedge^{p-1}(\calM)$ such that~
$
   L(\omega) = \rd_s \mu.
$

\end{thm}
\begin{remark}
\begin{enumerate}
\item
\label{rem:form} If $p=0$, then $\omega=f\in R$ and $\rd_s f=0$ if and only if $s=0$ or $\pa_{x_i}(f)=0$ for all $1\leq i \leq s$ if $s>0$. Therefore Corollary~\ref{COR:compatible} is a special case of Theorem~\ref{THM:ppl}.
\item Note that the parametrized Poincar\'{e} lemma is just the special case of Theorem~\ref{THM:ppl} when $s=n$.
\end{enumerate}
\end{remark}
\begin{proof}
We proceed by induction on~$s$. Assume that $s=0$ and write
$$
   \omega=\sum_{i=1}^m f_i \frakm_i
$$
where $\frakm_i$ a monomial in $\rd x_1, \rd x_2,\cdots, \rd x_n$ and $f_i\in R$. By Corollary~\ref{COR:compatible} with $I=\emptyset$, one can compute a nonzero $L\in k\langle \{t,\pa_t\}\cup V\rangle$ such that
$
   L(f_i)=0
$ for all $i=1,\cdots,m$.
Then one has that
$$
L(\omega)=\sum_{i=1}^m L(f_i)\frakm_i=0.
$$
This proves the base case. Now assume that the theorem holds for $s<\ell$ and consider the case $s=\ell$. Write
$$
   \omega=u\wedge \rd x_\ell + v
$$
where both $u$ and $v$ do not involve $\rd x_\ell$. Then the assumption $\rd_\ell \omega=0$ implies that
$$
  \rd_{\ell-1}u\wedge \rd x_\ell+\rd_\ell v=\rd_{\ell-1}u\wedge \rd x_\ell+\rd_{\ell-1} v+ \rd x_\ell \wedge \pa_{x_l}(v)=0.
$$
Since all of $\rd_{\ell-1}u, \rd_{\ell-1}v, \pa_{x_\ell}(v)$ do not involve $\rd x_\ell$, one has that
$\rd_{\ell-1} v=0$ and $\rd_{\ell-1}(u)-\pa_{x_\ell}(v)=0$. By the induction hypothesis, one can compute a nonzero $\tilde{L}\in k\langle \{t, x_\ell, \pa_t\}\cup V\rangle$ and $\tilde{\mu}\in \bigwedge^{p-1}(\calM)$ such that
\begin{equation}
\label{EQ:reduction1}
    \tilde{L}(v)= \rd_{\ell-1}(\tilde{\mu}).
\end{equation}
We claim that $\tilde{L}$ can be chosen to be free of $x_\ell$.
Write
$$\tilde{L}=\sum_{j=0}^d N_j x_\ell^d$$
where $N_j\in k\langle \{t,\pa_t\}\cup V\rangle$ and $N_d\neq 0$.
Multiplying $\tilde{L}$ by $\pa_{x_\ell}^d$ on the left and using the formula (\ref{EQ:formula2}) yield that
\begin{equation}
\label{EQ:reduction2}
   \pa_{x_\ell}^d \tilde{L}=\sum_{j=0}^d N_j\pa_{x_\ell}^d x_\ell^j=\alpha N_d+\tilde{N}\pa_{x_\ell}
\end{equation}
where $\alpha$ is a nonzero integer and $\tilde{N}\in k\langle \{t, x_\ell,\pa_t,\pa_{x_\ell}\}\cup V\rangle$. The equalities (\ref{EQ:reduction1}) and (\ref{EQ:reduction2}) together with $\pa_{x_\ell}(v)=\rd_{\ell-1}(\tilde{u})$ yield that
$
   N_d(v)=\rd_{\ell-1}(\pi)
$ for some $\pi\in \bigwedge^{p-1}(\calM)$.
This proves the claim. Now one has that
$$
   \tilde{L}(\omega)=\tilde{L}(u)\wedge \rd x_\ell + \rd_{\ell-1}(\tilde{\mu})=\tilde{L}(u)\wedge \rd x_\ell+\rd x_\ell\wedge \pa_{x_\ell}(\tilde{\mu})+\rd_\ell(\tilde{\mu}).
$$
Since $\tilde{L}$ is free of $x_1,\cdots,x_\ell$, $\tilde{L}\rd_\ell=\rd_\ell \tilde{L}$. This implies that
\begin{align*}
    0=\tilde{L}(\rd_\ell (\omega))=\rd_\ell(\tilde{L}(\omega))&=\rd_{\ell-1}(\tilde{L}(u))\wedge \rd x_\ell+ \rd x_\ell \wedge \rd_{\ell-1}(\pa_{x_\ell}(\tilde{\mu}))\\
    &=\rd_{\ell-1}\left(\tilde{L}(u)-\pa_{x_\ell}(\tilde{\mu})\right)\wedge \rd x_\ell.
\end{align*}
Note that $\tilde{\mu}$ can always be chosen to be free of $\rd x_\ell$. Hence one has that $\rd_{\ell-1}(\tilde{L}(u)-\pa_{x_\ell}(\tilde{\mu}))=0$. By the induction hypothesis, one can compute a nonzero $\bar{L}\in k\langle \{t, \pa_{x_\ell},\pa_t\}\cup V\rangle$ and $\bar{\mu}\in \bigwedge^{p-1}(\calM)$ such that
\begin{equation}
\label{EQ:reduction3}
\bar{L}\left(\tilde{L}(u)-\pa_{x_\ell}(\tilde{\mu})\right)=\rd_{\ell-1}(\bar{\mu}).
\end{equation}
Write
$$
   \bar{L}=\sum_{j=e_1}^{e_2}  \pa_{x_\ell}^j M_j
$$
where $M_j\in k\langle \{t, \pa_t\}\cup V\rangle$ and $M_{e_1}\neq 0$. Multiplying $\bar{L}$ by $x_\ell^{e_1}$ on the left and using the formula (\ref{EQ:formula2}) yield that
$$
   x_\ell^{e_1} \bar{L}=\beta M_{e_1}+ \pa_{x_\ell}\tilde{M}
$$
where $\beta$ is a nonzero integer and $\tilde{M}\in k\langle \{t,\pa_t,\pa_{x_\ell}, x_\ell\}\cup V\rangle$. Hence applying $x_\ell^{e_1}$ to the equality (\ref{EQ:reduction3}), one gets that
$$
   \beta M_{e_1}\left(\tilde{L}(u)-\pa_{x_\ell}(\tilde{\mu})\right)=\rd_{\ell-1}(x_\ell^{e_1}\bar{\mu})+\pa_{x_\ell}\left(\tilde{M}\left(\tilde{L}(u)-\pa_{x_\ell}(\tilde{\mu})\right)\right).
$$
Set $L=\beta M_{e_1}\tilde{L}$.
The one has that
\begin{align*}
     L(\omega)&=\beta M_{e_1}\left((\tilde{L}(u)-\pa_{x_\ell}(\tilde{\mu}))\wedge \rd x_\ell+\rd_\ell (\tilde{\mu})\right)\\
     &=\left(\beta M_{e_1}\left(\tilde{L}(u)-\pa_{x_\ell}(\tilde{\mu}\right)\right)\wedge \rd x_\ell +\rd_\ell(\beta M_{e_1}(\tilde{\mu}))\\
     &=\rd_{\ell-1}(x_\ell^{e_1}\bar{\mu})\wedge \rd x_\ell +\pa_{x_\ell}\tilde{M}\left(\tilde{L}(u)-\pa_{x_\ell}(\tilde{\mu})\right)\wedge \rd x_\ell + \rd_\ell(\beta M_{e_1}(\tilde{\mu}))\\
     &=\rd_\ell\left(x_\ell^{e_1}\bar{\mu}+\tilde{M}\left(\tilde{L}(u)-\pa_{x_\ell}(\tilde{\mu})\right)+\beta M_{e_1}(\tilde{\mu})\right).
\end{align*}
The last equality holds because
$$\rd_{\ell-1}\left(\tilde{M}\left(\tilde{L}(u)-\pa_{x_\ell}(\tilde{\mu})\right)\right)=\tilde{M}\rd_{\ell-1}\left(\tilde{L}(u)-\pa_{x_\ell}(\tilde{\mu})\right)=0.$$
\end{proof}

%By definition, we have the relation
%\[\rd = \rd_{n-1} + \rd^n.\]
%Let~$\omega \in \bigwedge(\calM)$, we can always decompose~$\omega$ into
%\[\omega = \omega_{n-1} + \omega^n,\]
%where $\omega_{m-1}\in \bigwedge_{m-1}(\calM)$ and $\omega^m = \mu \wedge \rd x_m$ with $\mu \in \bigwedge_{m-1}(\calM)$. Then we have the following lemma which can be proved easily according to the definition.
%
%\begin{lem}
%\label{LM:form}
%Let~$\omega = \omega_{m-1} + \omega^m  \in \bigwedge(\calM)$. we have
%\begin{itemize}
%\item[(i)]~$\rd^m \omega^m = 0$;
%\item[(ii)]~$\rd(\omega\wedge \rd x_m) = \rd_{m-1}\omega\wedge \rd x_m$;
%\item[(iii)]~$\rd\omega =0$ if and only if~$\rd_{m-1} \omega_{m-1} = 0$ and~$\rd^m\omega_{m-1}  + \rd_{m-1}\omega^m=0$.
%\end{itemize}
%\end{lem}

\begin{remark}
Lemma~\ref{LM:ppl} can be derived from the finiteness of the de Rham cohomology groups of $D$-modules in the Bernstein class. To see this, let $\omega$  be a differential $s$-form with coefficients in $R$ and let $M$ be the $D$-module generated by all coefficients of $\omega$ and all derivatives of these coefficients with respect to $\pa_t$.   By Proposition 5.2 on page 12 of \cite{bjork}, $M$ is a $D$-module in the Bernstein class. Assume that $\omega$ is closed. Then $\pa_t^j(\omega)\in H_{DR}^s(M)$, the $j$-th de Rham cohomology group of $M$, for all nonnegative integer $j$. By Theorem 6.1 on page 16 of \cite{bjork}, $H_{DR}^s(M)$ is of finite dimension over $k(t)$. This implies that there are $a_0,\cdots, a_m\in k(t)$ such that $\sum_{j=0}^m a_j \pa_t^j(\omega)=0$ in $H_{DR}^s(M)$, i.e. $\sum_{j=0}^m a_j \pa_t^j(\omega)$ is exact. This proves the existence of telescopers for the $\pa_t$-closed differential forms. However the proof of Theorem~\ref{THM:ppl} is constructive and it provides a method to compute a telescoper if it exists.
\end{remark}
The proof of Theorem~\ref{THM:ppl} can be summarized as the following algorithm.
\begin{algorithm}
\label{ALG:telescopers}
Input: $\omega\in \bigwedge^{p}(\calM)$ and $V\in \{x_i,\pa_{x_i}|i=s+1,\cdots,n\}$  satisfying that $\rd_s(\omega)=0$ and $|V|=n-s$ \\
Output: a nonzero $L\in k\langle \{t,\pa_t\}\cup V\rangle$ such that $L(\omega)=\rd_s(\mu)$.
\begin{enumerate}
\item  If $\omega\in R$, then by Corollary~\ref{COR:compatible}, compute a nonzero $L\in k\langle \{t, \pa_t\}\cup V \rangle$ such that $L(\omega)=0$. Return $L$.
\item  Write $\omega=u\wedge \rd x_s + v$ with $u,v$ not involving $\rd x_s$.
\item  Call Algorithm~\ref{ALG:telescopers} with $v$ and $V\cup \{x_s\}$ as inputs and let $\tilde{L}$ be the output.
\begin{enumerate}
\item Write $\tilde{L}=\sum_{j=0}^d N_j x_s^j $ with $N_j\in k\langle \{t,\pa_t, x_s\}\cup V\rangle$ and $N_d\neq 0$.
\item Compute a $\tilde{\mu}\in \bigwedge^{p-1}(\calM)$ such that $N_d(v)=\rd_{s-1}(\tilde{\mu})$.
\end{enumerate}
\item Write $N_d(\omega)=(N_d(u)-\pa_{x_s}(\tilde{\mu}))\wedge \rd x_s+ \rd_s(\tilde{\mu})$.
\item Call Algorithm~\ref{ALG:telescopers} with $N_d(u)-\pa_{x_s}(\tilde{\mu})$ and $V\cup \{\pa_{x_s}\}$ as inputs and let $\bar{L}$ be the output.
\item Write $\bar{L}=\sum_{j=e_1}^{e_2} \pa_{x_s}^j M_j$ with $M_j\in k\langle \{t,\pa_t\}\cup V\rangle$ and $M_{e_1}\neq 0$.
\item Return $M_{e_1}N_d$.
\end{enumerate}
\end{algorithm}
%\begin{example}
%Let $\omega$ be as in Example~\ref{EX:calabi-yau}. An easy calculation implies that $\omega$ is closed. We shall compute a telescoper for it.
%\begin{enumerate}
%\item Write $\omega=u\wedge \rd x_5 +v$, where
%$$u=\sum_{i=1}^4 \frac{(-1)^{i-1}x_i}{W(t)} \rd x_1 \wedge\cdots \wedge \widehat{\rd x_i} \wedge\cdots \wedge \rd x_4,\,\,v=\frac{x_5}{W(t)}\rd x_1\wedge \cdots \wedge \rd x_4.$$
%\end{enumerate}
%\end{example}
\section{The existence of telescopers}
\label{sec:existence}
It is easy to see that if a differential form is $\pa_t$-exact then it is $\pa_t$-closed.
Therefore Lemma~\ref{LM:ppl} implies that given a $\omega\in \bigwedge^p(\calM)$, to decide whether it has a telescoper, it suffices to decide whether there is a nonzero $L\in k\langle t, \pa_t \rangle$ such that $L(\rd\omega)=0$. Suppose that
$$\rd \omega=\sum_{1\leq i_1<\dots<i_{p+1}\leq n} a_{i_1,\dots,i_{p+1}}\rd x_{i_1}\cdots\rd x_{p+1}, \,\,a_{i_1,\dots,a_{p+1}}\in \calU.$$
Then $L(\rd \omega)=0$ if and only if $L(a_{i_1,\dots,i_{p+1}})=0$ for all $1\leq i_1<\cdots<i_{p+1}\leq n$.
So the existence problem of telescopers can be reduced to the following problem.
\begin{problem}
\label{prob1}
Given an element $f\in \calU$, decide whether there exists a nonzero $L\in k\langle t,\pa_t\rangle$ such that $L(f)=0$.
\end{problem}
Let $P\in K\langle \pa_t \rangle\setminus\{0\}$ be the monic operator of minimal order such that $P(f)=0$. Then $f$ is annihilated by a nonzero $L\in k(t)\langle \pa_t \rangle$ if and only if $P$ is a right-hand factor of $L$, i.e. $L=QP$ for some $Q\in K\langle \pa_t\rangle$.  Such operator $P$ will be  called a  $(\bfx,t)$-separable operator.
 Problem~\ref{prob1} then is equivalent to the following one.
\begin{problem}
\label{prob2}
Given a $P\in K\langle \pa_t \rangle\setminus\{0\}$, decide whether $P$ is $(\bfx,t)$-separable.
\end{problem}
The rest of this paper is aimed at developing an algorithm to solve the above problem. Let us first investigate the solutions of  $(\bfx,t)$-separable operators.
\begin{notation}
$$C_t:=\left\{ c\in \calU \mid \pa_t(c)=0\right\},\,\,C_\vx:=\left\{ c\in \calU \mid \forall\, x\in \vx, \pa_x(c)=0\right\}.$$
\end{notation}
Assume that  $L\in k(t)\langle \pa_t \rangle\setminus\{0\}$. By Corollary 1.2.12 of \cite{singer}, the solution space of $L=0$ in $\calU$ is a $C_t$-vector space of dimension $\ord(L)$. Moreover we have the following lemma.
\begin{lem}
\label{LM:solutions}
If $L\in k(t)\langle \pa_t \rangle\setminus\{0\}$, then the solution space of $L=0$ in $\calU$ has a basis in $C_\vx$.
\end{lem}
\begin{proof}
Let $d=\ord(L)>0$ and $\{v_1,\cdots,v_d\}$ be a basis of the solution space of $L=0$ in $\calU$. For all $1\leq i \leq d$ and all $1\leq l \leq m$,
$$
     L(\pa_{x_l}(v_i))=\pa_{x_l}(L(v_i))=0.
$$
Set $\vv=(v_1,\cdots,v_d)^t$. Then for each $l=1,\cdots,m$,
$$
    \pa_{x_l}(\vv)=A_l \vv, \,\, A_l\in \Mat_d(C_t).
$$
Since $\pa_{x_i}\pa_{x_j}(\vv)=\pa_{x_j}\pa_{x_i}(\vv)$ and $v_1,\cdots,v_d$ are linearly independent over $C_\vx$, for all $1\leq i < j \leq m$,
\begin{equation}
\label{EQ:compatible}
     \pa_{x_i}(A_j)-\pa_{x_j}(A_i)=A_iA_j-A_jA_i
\end{equation}
On the other hand, $\pa_t(A_i)=0$ for all $1\leq i \leq d$. These together with (\ref{EQ:compatible}) imply that
the system
$$\pa_{x_1}(Y)=A_1Y, \cdots, \pa_{x_m}(Y)=A_mY, \pa_t(Y)=0$$
is integrable. Then there is an invertible matrix $G$ with entries in $\calU$ satisfying this system.  Let $\bar{\vv}=G^{-1}\vv$. As $\pa_t(G^{-1})=0$, $\bar{\vv}$ is still a basis of the solution space of $L=0$ in $\calU$.  Furthermore, for each $i=1,\cdots,m$, we have
$$
     \pa_{x_i}(\bar{\vv})=\pa_{x_i}(G^{-1}\vv)=\pa_{x_i}(G^{-1})\vv+G^{-1}A_i\vv=-G^{-1}A_i\vv+G^{-1}A_i\vv=0.
$$
Thus $\bar{\vv}\in C_\vx^d$.
\end{proof}
As a consequence, we have the following corollary.
\begin{cor}
\label{COR:solutions}
Assume that $P\in K\langle \pa_t \rangle\setminus\{0\}$. Then $P$ is $(\bfx,t)$-separable if and only if the solutions of $P(y)=0$ in $\calU$ are of the form
\begin{equation}
\label{EQ:removableform}
     \sum_{i=1}^s g_i h_i,\,\, g_i\in C_t, h_i\in C_\vx\cap \{f\in \calU \mid P(f)=0\}.
\end{equation}
\end{cor}
\begin{proof}
The ``only if" part is a direct consequence of Lemma~\ref{LM:solutions}. For the ``if" part, one only need to prove that if $h\in C_\vx\cap \{f\in \calU \mid P(f)=0\}$ then $h$ is annihilated by a nonzero operator in $k(t)\langle \pa_t \rangle$. Suppose that $h\in C_\vx\cap\{f\in \calU \mid P(f)=0\}$. Let $L$ be the monic operator in $K\langle \pa_t \rangle\setminus \{0\}$ which annihilates $h$ and is of minimal order. Write
$$
   L=\pa_t^\ell+\sum_{i=0}^{\ell-1} a_i \pa_t^i, a_i\in K.
$$
Then for every $j\in \{1,\dots,m\}$
$$
  0=\pa_{x_j}(L(h))=\sum_{i=0}^{\ell-1} \pa_{x_j}(a_i) \pa_t^i(h)+L(\pa_{x_j}(h))=\sum_{i=0}^{\ell-1} \pa_{x_j}(a_i) \pa_t^i(h).
$$
The last equality holds because $h\in C_\bfx$. By the miniality of $L$, one sees that $\pa_{x_j}(a_i)=0$ for all $i=0,\dots,\ell-1$ and all $j=1,\dots,m$. Hence $a_i\in k(t)$ for all $i$. In other words, $L\in k(t)\langle \pa_t \rangle$.
\end{proof}
For convention, we introduce the following definition.
\begin{definition}
\begin{itemize}
\item [$(1)$]
We say $f\in \calU$ is  split if it can be written as the form $f=gh$ where $g\in C_t$ and $h\in C_\vx$, and say $f$ is semisplit if it is the sum of finitely many split elements.
\item [$(2)$] We say a nonzero operator $P\in K\langle \pa_t \rangle$ is semisplit if it is monic and all its coefficients are semisplit.
\end{itemize}
\end{definition}
The semisplit operators have the following property.
\begin{lem}
\label{LM:samecoefficients}
Assume that $P=Q_1Q_2$ where $P,Q_1,Q_2$ are monic operators in $K\langle \pa_t \rangle$. Assume further that $Q_2\in k(t)[\bfx,1/r]\langle \pa_t \rangle$ where $r\in k[\bfx,t]$. Then $P\in k(t)[\bfx,1/r]\langle \pa_t \rangle$ if and only if so is $Q_1$.
\end{lem}
\begin{proof}
Comparing the coefficients on both sides of $P=Q_1Q_2$ concludes the lemma.
\end{proof}
As a direct consequence, we have the following corollary.
\begin{cor}
\label{cor:semisplitoperators}
Assume that $P=Q_1Q_2$ where $P,Q_1,Q_2$ are monic operators in $K\langle \pa_t \rangle$. Assume further that $Q_2$ is semisplit. Then $P$ is semisplit if and only if so is $Q_1$.
\end{cor}

\subsection{The completely reducible case}
In Proposition 10 of \cite{chen-feng-li-singer}, we show that given a hyperexponential function $h$ over $K$, $\ann(h)\cap k(t)\langle \pa_t \rangle\neq \{0\}$ if and only if there is a nonzero $p\in k(\bfx)[t]$ and $r\in k(t)$ such that
$$
   a=\frac{\pa_t(p)}{p}+r,
$$
where $a=\pa_t(h)/h$. Remark that $a,p, r$ with $p\neq 0$ satisfy the above equality if and only if $\frac{1}{p}(\pa_t-a)=(\pa_t-r)\frac{1}{p}$. Under the notion of $(\bfx,t)$-separable and the language of differential operators, Proposition 10 of \cite{chen-feng-li-singer} states that $\pa_t-a$ is $(\bfx,t)$-separable if and only if it is similar to a first order operator in $k(t)\langle \pa_t\rangle$ by some $1/p$ with $p$ being nonzero polynomial in $t$.  In this section, we shall generalize Proposition 10 of \cite{chen-feng-li-singer} to the case of completely reducible operators.
We shall use $\lclm(Q_1,Q_2)$ to denote the monic operator of minimal order which is divisible by both $Q_1$ and $Q_2$ on the right.
We shall prove that if $P$ is $(\bfx,t)$-separable and completely reducible then there is a nonzero $L\in k(t)\langle \pa_t\rangle$ such that $P$ is the transformation of $L$ by  some $Q$ with semisplit coefficients.
To this end, we need to introduce some notations from \cite{ore}.

\begin{definition}
Assume that $P, Q\in K\langle \pa_t \rangle \setminus\{0\}$.
\begin{enumerate}
\item We say $\tilde{P}$ is the transformation of $P$ by $Q$ if $\tilde{P}$ is the monic operator satisfying that $\tilde{P}Q=\lambda \lclm(P,Q)$ for some $\lambda\in K$.
\item We say $\tilde{P}$ is similar to $P$ (by $Q$) if there is an operator $Q$ with $\gcrd(P,Q)=1$ such that $\tilde{P}$ is the transformation of $P$ by $Q$, where $\gcrd(P,Q)$ denotes the greatest common right-hand factor of $P$ and $Q$.
\end{enumerate}
\end{definition}
\begin{definition}
\begin{enumerate}
\item
We say $P\in K\langle \pa_t \rangle$ is completely reducible if it is the lclm of a family of irreducible operators in $K\langle \pa_t \rangle$.
\item
We say $Q\in K\langle \pa_t \rangle$ is the maximal completely reducible right-hand factor of $P\in K\langle \pa_t\rangle$ if $Q$ is the lclm of all irreducible right-hand factros of $P$.
\end{enumerate}
\end{definition}
Given a $P\in K\langle \pa_t \rangle$, Theorem 7 of \cite{ore} implies that $P$ has the following unique decomposition called the maximal completely reducible decomposition or the m.c.r. decomposition for short,
$$
   P=\lambda H_r H_{r-1} \dots H_1
$$
where $\lambda\in K$ and $H_i$ is the maximal completely reducible right-hand factor of $H_r \dots H_i$. For an $L\in k(t)\langle \pa_t \rangle$, it has two m.c.r. decompositions viewed it as an operator in $k(t)\langle \pa_t \rangle$ and an operator in $K\langle \pa_t \rangle$ respectively. In the following, we shall prove that these two decompositions coincide.
For convenience, we shall denote by $P_{x_i=c_i}$ the operator obtained by replacing $x_i$ by $c_i\in k$ in $P$.
\begin{lem}
\label{LM:gcrd}
Assume that $P, L$ are two monic operators in $K\langle \pa_t \rangle$. Assume further that $P\in k(t)[\bfx,1/r]\langle \pa_t \rangle$ with $r\in k[\bfx,t]$, and $L\in k(t)\langle \pa_t \rangle$. Let $\vc\in k^m$ be such that $r(\vc)\neq 0$.
\begin{enumerate}
\item
If $\gcrd(P_{\bfx=\bfc}, L)=1$ then $\gcrd(P,L)=1$.
\item If $\gcrd(P,L)=1$ then there is $\va\in k^m$ such that $r(\va)\neq 0$ and $\gcrd(P_{\bfx=\va},L)=1$.
\end{enumerate}
\end{lem}
\begin{proof}
1. We shall prove the lemma by induction on $m=|\bfx|$. Assume that $m=1$, and $\gcrd(P,L)\neq 1$. Then there are $M,N\in k(t)[x_1]\langle \pa_t \rangle$ with $\ord(M)<\ord(L)$ such that
$
   MP+NL=0.
$
Write $$M=\sum_{i=0}^{n-1} a_i\pa_t^i, \quad N=\sum_{i=0}^s b_i\pa_t^i$$
where $n=\ord(L)$. If the $a_i$'s have a common factor $c$ in $k(t_1)[x_1]$, then one sees that $c$ is a common factor of the $b_i$'s. Thus we can cancel this factor $c$. So without loss of generality, we may assume that the $a_i$'s have no common factor. This implies that $M_{x_1=c_1}\neq 0$ and $M_{x_1=c_1}P_{x_1=c_1}+N_{x_1=c_1}L=0$. Since $\ord(M_{x_1=c_1})<\ord(L)$, $\gcrd(P_{x_1=c_1}, L)\neq 1$, a contradiction. For the general case, set $Q=P_{x_1=c_1}$. Then $Q_{x_2=c_2,\dots,x_m=c_m}=P_{\bfx=\vc}$. This implies that $\gcrd(Q_{x_2=c_2,\dots,x_m=c_m},L)=1$. By the induction hypothesis, $\gcrd(Q,L)=1$. Finally, regarding $P$ and $L$ as operators with coefficients in $k(t,x_2,\dots,x_m)[x_1,1/r]$ and by the induction hypothesis again, we get $\gcrd(P,L)=1$.

2. Since $\gcrd(P,L)=1$, there are $M,N\in K\langle \pa_t \rangle$ such that $MP+NL=1$. Let $\va\in k^m$ be such that $r(\va)\neq 0$ and both $M_{\bfx=\va}$ and $N_{\bfx=\va}$ are well-defined. For such $\va$, one has that $M_{\bfx=\va}P_{\bfx=\va}+N_{\bfx=\va}L=1$ and then $\gcrd(P_{\bfx=\va},L)=1$.
\end{proof}
\begin{lem}
\label{LM:mcrdecomposition}
Let $L\in k(t)\langle \pa_t \rangle$. The m.c.r. decompositions of $L$ viewed as an operator in $k(t)\langle \pa_t \rangle$ and an operator in $K\langle \pa_t \rangle$ respectively coincide.
\end{lem}
\begin{proof}
We first claim that an irreducible operator of $k(t)\langle \pa_t \rangle$ is irreducible in $K\langle \pa_t \rangle$. Let $P$ be a monic irreducible operator in $k(t)\langle \pa_t \rangle$ and assume that $Q$ is a monic right-hand factor of $P$ in $K\langle \pa_t \rangle$ with $1\leq \ord(Q)<\ord(P)$. Then $P=\tilde{Q}Q$ for some $\tilde{Q}\in K\langle \pa_t \rangle$. Suppose that $Q\in k(t)[\bfx,1/r]\langle \pa_t \rangle$. By Lemma~\ref{LM:samecoefficients}, $\tilde{Q}$ belongs to $k(t)[\bfx,1/r]\langle \pa_t \rangle$. Let $\vc\in k^m$ be such that $r(\vc)\neq 0$. Then $P=\tilde{Q}_{\bfx=\vc}Q_{\bfx=\vc}$ and $1\leq \ord(Q_{\bfx=\vc})\leq \ord(P)$. These imply that $P$ is reducible, a contradiction. So $P$ is irreducible and thus the claim holds.
Let $L=\lambda H_r H_{r-1}\dots H_1$ be the m.c.r. decomposition in $k(t)\langle \pa_t \rangle$. The above claim implies that $H_1$ viewed as an operator in $K\langle \pa_t \rangle$ is completely reducible. Assume that $H_1$ is not the maximal compleltely reducible right-hand factor of $L$ in $K\langle \pa_t \rangle$. Let $M\in K\langle \pa_t \rangle\setminus K$ be a monic irreducible right-hand factor of $L$ satisfying that $\gcrd(M,H_1)=1$. Due to Lemma~\ref{LM:gcrd}, there is $\va\in k^m$ satisfying that $\gcrd(M_{\bfx=\va},H_1)=1$. Note that $M_{\bfx=\va}$ is a right-hand factor of $L$. Therefore $M_{\bfx=\va}$ has some irreducible right-hand factor of $L$ as a right-hand factor. Such irreducible factor must be a right-hand factor of $H_1$ and thus $\gcrd(M_{\bfx=\va}, H_1)\neq 1$, a contradiction. Therefore $H_1$ is the maximal completely reducible right-hand factor of $L$ in $K\langle \pa_t \rangle$. Using the induction on the order, one sees that $\lambda H_r H_{r-1}\dots H_1$ is the m.c.r. decomposition of $L$ in $K\langle \pa_t \rangle$.
\end{proof}
\begin{lem}
\label{LM:similarity}
Assume that $P$ is monic, $(\bfx,t)$-separable and completely reducible. Assume further that $P\in k(t)[\bfx,1/r]\langle \pa_t \rangle$ with $r\in k[\bfx,t]$. Let $\vc\in k^m$ be such that $r(\vc)\neq 0$. Then $P_{\bfx=\bfc}$ is similar to $P$.
\end{lem}
\begin{proof}
 Let $\tilde{L}$ be a nonzero monic operator in $k(t)\langle \pa_t \rangle$ with $P$ as a right-hand factor. Since $P$ is completely reducible, by Theorem 8 of \cite{ore}, $P$ is a right-hand factor of the maximal completely reducible right-hand factor of $\tilde{L}$. By Lemma~\ref{LM:mcrdecomposition},  the maximal completely reducible right-hand factor of $\tilde{L}$ is in $k(t)\langle \pa_t \rangle$.  Hence we may assume that $\tilde{L}$ is completely reducible after replacing $\tilde{L}$ by its maximal completely reducible right-hand factor. Assume that $\tilde{L}=QP$ for some $Q\in K\langle \pa_t \rangle$. By Lemma~\ref{LM:samecoefficients}, $Q\in k(t)[\bfx,1/r]\langle \pa_t \rangle$. Then $\tilde{L}=Q_{\bfx=\bfc}P_{\bfx=\bfc}$, i.e. $P_{\bfx=\bfc}$ is a right-hand factor of $\tilde{L}$. We claim that for a right-hand factor $T$ of $\tilde{L}$, there is a right-hand factor $L$ of $\tilde{L}$ satisfying that $\gcrd(T,L)=1$ and $\lclm(T,L)=\tilde{L}$. We prove this claim by induction on $s=\ord(\tilde{L})-\ord(T)$. When $s=0$, there is nothing to prove. Assume that $s>0$. Then since $\tilde{L}$ is completely reducible, there is an irreducible right-hand factor $L_1$ of $\tilde{L}$ such that $\gcrd(T,L_1)=1$. Let $N=\lclm(T,L_1)$. We have that $\ord(N)=\ord(T)+\ord(L_1)$. Therefore $\ord(\tilde{L})-\ord(N)<s$. By induction hypothesis, there is a right-hand factor $L_2$ of $\tilde{L}$ such that $\gcrd(N,L_2)=1$ and $\lclm(N,L_2)=\tilde{L}$. Let $L=\lclm(L_1,L_2)$. Then
 $$
   \tilde{L}=\lclm(N,L_2)=\lclm(T, L_1,L_2)=\lclm(T,L).
 $$
Taking the order of the operators in the above equality yields that
 \begin{align*}
   \ord(\lclm(T,L))&=\ord(\lclm(N,L_2))=\ord(N)+\ord(L_2)\\
     &=\ord(T)+\ord(L_1)+\ord(L_2).
 \end{align*}
 On the other hand, we have
$$
   \ord(\lclm(T,L))\leq \ord(T)+\ord(L) \leq \ord(T)+\ord(L_1)+\ord(L_2).
$$
This implies that
$$
    \ord(\lclm(T,L))= \ord(T)+\ord(L).
$$
So $\gcrd(T,L)=1$ and then $L$ is a required operator. This proves the claim. Now let $L_{\vc}$ be a ritht-hand factor of $\tilde{L}$ satisfying that $\gcrd(P_{\bfx=\vc}, L_\vc)=1$ and $\lclm(P_{\bfx=\vc}, L_\vc)=\tilde{L}$. Let $M\in k(t)\langle \pa_t \rangle$ be such that $\tilde{L}=ML_\vc$. Then $P_{\bfx=\bfc}$ is similar to $M$. It remains to show that $P$ is also similar to $M$. Due to Lemma~\ref{LM:gcrd}, $\gcrd(P,L_\vc)=1$. Then
 $$\ord(\lclm(P,L_\vc))=\ord(P)+\ord(L_\vc)=\ord(P_{\bfx=\bfc})+\ord(L_\vc)=\ord(\tilde{L}).$$
 Note that $\lclm(P,L_\vc)$ is a right-hand factor of $\tilde{L}$. Hence $\lclm(P,L_\vc)=\tilde{L}$ and thus $P$ is similar to $M$.
\end{proof}

For the general case, the above lemma is not true anymore as shown in the following example.
\begin{example}
Let $y=x_1\log(t+1)+x_2\log(t-1)$ and
$$P=\pa_t^2+\frac{(t-1)^2x_1+(t+1)^2x_2}{(t^2-1)((t-1)x_1+(t+1)x_2)}\pa_t.$$
Then $P$ is $(x,t)$-separable since $\{1,y\}$ is a basis of the solution space of $P=0$ in $\calU$. We claim that $P$ is not similar to $P_{\bfx=\bfc}$ for any $\bfc\in k^2\setminus\{(0,0)\}$. Suppose on the contrary that $P$ is similar to $P_{\bfx=\bfc}$ for some $\bfc=(c_1,c_2)\in k^2\setminus\{(0,0)\}$, i.e. there are $a,b\in k(\bfx,t)$, not all zero, such that $\gcrd(a\pa_t+b, P_{\bfx=\bfc})=1$ and $P$ is the transformation of $P_{\bfx=\bfc}$ by $a\pa_t+b$. Denote $Q=a\pa_t+b$. As $\{1, y_{\bfx=\bfc}\}$ is a basis of the solution space of $P_{\bfx=\bfc}$, $\{Q(1), Q(y_{\bfx=\bfc})\}$ is a basis of the solution space of $P=0$. In other words, there is $C\in \GL_2(C_t)$ such that
$$
    \left(b, a\left(\frac{c_1}{t+1}+\frac{c_2}{t-1}\right)+by_{\bfx=\bfc}\right)=(1,y)C.
$$
Note that $\log(t+1),\log(t-1),1$ are linearly independent over $k(x_1,x_2,t)$. We have that $b\in C_t\setminus\{0\}$ and $bc_1=\tilde{c}x_1, bc_2=\tilde{c}x_2$ for some $\tilde{c}\in C_t$. This implies that $x_1/x_2=c_1/c_2\in k$, a contradiction.
\end{example}

When the given two operators are of length two, i.e. they are the products of two irreducible operators, a criterion for the similarity is presented in \cite{li-wang}. For the general case, suppose that $P$ is similar to $P_{\bfx=\vc}$ by $Q$. Then the operator $Q$ is a solution of the following mixed differential equation
\begin{equation}\label{EQ:mixedequation}
    Pz\equiv 0\mod P_{\bfx=\vc}.
\end{equation}
An algorithm for computing all solutions of the above mixed differential equation is developed in \cite{vanhoeij1}.  In the following, we shall show that if $P$ is $(\bfx,t)$-separable then $Q$ is an operator with semisplit coefficients. Note that $Q$ can be chosen to be of order less than $\ord(P_{\bfx=\vc})$ and all solutions of the mixed differential equation with order less than $\ord(P_{\bfx=\vc})$ form a vector space over $k(\bfx)$ of finite dimension. Furthermore $Q$ induces an isomorphism from the solution space of $P_{\bfx=\vc}(y)=0$ to that of $P(y)=0$.
\begin{prop}
\label{PROP:criterion}
Assume that $P$ is monic and completely reducible. Assume further that $P\in k(t)[\bfx,1/r]\langle \pa_t \rangle$ with $r\in k[\bfx,t]$. Let $\vc\in k^m$ be such that $r(\vc)\neq 0$. Then $P$ is $(\bfx,t)$-separable if and only if $P$ is similar to $P_{\bfx=\vc}$ by an operator $Q$ with semisplit coefficients.
\end{prop}
\begin{proof}
Denote $n=\ord(P_{\bfx=\vc})=\ord(P)$.
Assume that $\{\alpha_1,\cdots,\alpha_n\}$ is a basis of the solution space of $P_{\bfx=\bfc}(y)=0$ in $C_\bfx$ and $P$ is similar to $P_{\bfx=\vc}$ by $Q$. Write $Q=\sum_{i=0}^{n-1} a_i \pa_t^i$ where $a_i\in K$. Then
$$
\left(Q(\alpha_1),\dots,Q(\alpha_n)\right)=(a_0,\dots,a_{n-1})\begin{pmatrix}\alpha_1 & \alpha_2 & \dots & \alpha_n \\
\alpha_1' & \alpha_2' & \dots & \alpha_n' \\
\vdots & \vdots & & \vdots\\
\alpha_1^{(n-1)} & \alpha_2^{(n-1)} & \dots & \alpha_n^{(n-1)}
\end{pmatrix}
$$
and $Q(\alpha_1),\dots,Q(\alpha_n)$ form a basis of the solution space of $P(y)=0$.

Now suppose that $P$ is $(\bfx,t)$-separable. Due to Lemma~\ref{LM:similarity}, $P$ is similar to $P_{\bfx=\vc}$ by $Q$.
By Corollary~\ref{COR:solutions}, the $Q(\alpha_i)$ are semisplit. The above equalities then imply that the $a_i$ are semisplit. Conversely, assume that $P$ is similar to $P_{\bfx=\vc}$ by $Q$ and the $a_i$ are semisplit.
It is easy to see the $Q(\alpha_i)$ are semisplit. By Corollary~\ref{COR:solutions} again, $P$ is $(\bfx,t)$-separable.
\end{proof}

Using the algorithm developed in \cite{vanhoeij1}, we can compute a basis of the solution space over $k(\bfx)$ of the equation (\ref{EQ:mixedequation}). It is clear that the solutions with semisplit entries form a subspace. We can compute a basis for this subspace as follows. Suppose that $\{Q_1,\dots,Q_\ell\}$ is a basis of the solution space of the equation (\ref{EQ:mixedequation}) consisting of solutions with order less than $\ord(P_{\bfx=\vc})$. We may identity $Q_i$ with a vector $\vg_i\in K^n$ under the basis $1,\pa_t,\dots,\pa_t^{n-1}$. Let $q\in k(\bfx)[t]$ be a common denominator of all entries of the $\vg_i$. Write $\vg_i=\vp_i/q$ for each $i=1,\dots,\ell$, where $\vp_i\in k(\bfx)[t]^n$. Write $q=q_1 q_2$ where $q_2$ is split but $q_1$ is not. Note that a rational function in $t$ with coefficients in $k(\bfx)$ is semisplit if and only if its denominator is split. For $c_1,\dots, c_\ell\in k(\bfx)$, $\sum_{i=1}^\ell c_i \vg_i$ is semisplit if and only if all entries of $\sum_{i=1}^\ell c_i \vp_i$ are divided by $q_1$. For $i=1,\dots, \ell$, let $\vh_i$ be the vector whose entries are the remainders of the corresponding entries of $\vp_i$ by $q_1$. Then all entries of $\sum_{i=1}^\ell c_i \vp_i$ are divided by $q_1$ if and oly if $\sum_{i=1}^\ell c_i\vh_i=0$. Let $\vc_1,\dots,\vc_s$ be a basis of the solution space of $\sum_{i=1}^\ell z_i \vh_i=0$. Then $\{(Q_1,\dots,Q_\ell) \vc_i\mid i=1,\dots,s\}$ is the required basis. Consequently, the required basis can be computed by solving the system of linear equations $\sum_{i=1}^\ell z_i\vh_i=0$.

In the following, for the sake of notations, we assume that $\{Q_1,\dots,Q_\ell\}$ is a basis of the solution space of the equation (\ref{EQ:mixedequation}) consisting of solutions with semisplit coefficients. By Proposition~\ref{PROP:criterion} and the definition of similarity, $P$ is $(\bfx,t)$-separable if and only if there is a nonzero $\tilde{Q}$ in the space spanned by $Q_1,\dots,Q_\ell$ such that $\gcrd(P_{\bfx=\vc}, \tilde{Q})=1$. Note that $\tilde{Q}$ induces a homomorphism from the solutions space of $P_{\bfx=\vc}(y)=0$ to that of $P(y)=0$. Moreover, one can easily see that $\gcrd(P_{\bfx=\vc},\tilde{Q})=1$ if and only if $\tilde{Q}$ is an isomorphism i.e. $\tilde{Q}(\alpha_1),\dots,\tilde{Q}(\alpha_n)$ form a basis of the solution space of $P(y)=0$ where $\{\alpha_1,\dots,\alpha_n\}$ is a basis of the solution space of $P_{\bfx=\vc}(y)=0$. Assume that $\tilde{Q}=\sum_{i=0}^{n-1} a_{0,i} \pa_t^i$ with $a_{0,i}\in K$. Using the relation $P_{\bfx=\vc}(\alpha_j)=0$ with $j=1,\dots,n$, one has that for all $j=1,\dots,n$
$$
  \tilde{Q}(\alpha_j)'=\left(\sum_{i=0}^{n-1} a_{0,i} \alpha_j^{(i)}\right)'=\sum_{i=0}^{n-1} a_{1,i} \alpha_j^{(i)}
$$
for some $a_{1,i}\in K$. Repeating this process, we  can compute $a_{l,i}\in K$ such that for all $j=1,\dots,n$ and $l=1,\dots, n-1$,
$$
   \tilde{Q}(\alpha_j)^{(l)}=\sum_{i=0}^{n-1} a_{l,i} \alpha_j^{(i)}.
$$
Now suppose that $\tilde{Q}=\sum_{i=1}^\ell z_i Q_i$  with $z_i\in k(\bfx)$. One sees that the $a_{l,i}$ are linear in $z_1,\dots,z_\ell$. Set $A(\vz)=(a_{i,j})_{0\leq i,j\leq n-1}$ with $\vz=(z_1,\dots,z_\ell)$. Then one has that
\begin{equation}
\label{EQ:transformation}
A(\vz)\begin{pmatrix}
   \alpha_1 & \dots & \alpha_n \\
    \vdots & & \vdots \\
   \alpha^{(n-1)}& \dots & \alpha_n^{(n-1)}
 \end{pmatrix}=\begin{pmatrix}
                 \tilde{Q}(\alpha_1) & \dots & \tilde{Q}(\alpha_n) \\
                 \vdots &  & \vdots \\
                 \tilde{Q}(\alpha_1)^{(n-1)} & \dots & \tilde{Q}(\alpha_n)^{(n-1)}
               \end{pmatrix}.
\end{equation}
It is well-known that $\tilde{Q}(\alpha_1),\dots,\tilde{Q}(\alpha_n)$ form a basis if and only if the right-hand side of the above equality is a nonsingular matrix and thus if and only if $A(\vz)$ is nonsingular. In the sequel, one can reduce the problem of the existence of $\tilde{Q}$ satisfying $\gcrd(\tilde{Q},P_{\bfx=\vc})=1$ to the problem of the existence of $\va\in k(\bfx)^\ell$ in $k(\bfx)$ such that $\det(A(\va))\neq 0$.

Suppose now we have already had an operator $Q$ with semisplit coefficients such that $P$ is similar to $P_{\bfx=\vc}$ by $Q$. Write $Q=\sum_{i=0}^{n-1} b_i \pa_t^i$ where $b_i\in K$ is semisplit. Write further $b_i=\sum_{j=1}^s h_{i,j}\beta_j$ where $h_{i,j}\in k(\bfx)$ and $\beta_j\in k(t)\setminus \{0\}$. Let $L_0=P_{\bfx=\bfc}$ and let $L_i$ be the transformation of $L_{i-1}$ by $\pa_t$ for $i=1,\cdots,n-1$. Then $L_i$ annihilates $\alpha_j^{(i)}$ for all $j=1,\cdots,n$ and $L_i \frac{1}{\beta_l}$ annihilates $\beta_l \alpha_j^{(i)}$ for all $l=1,\dots,s$ and $j=1,\dots,n$. Set
$$L=\lclm\left(\left\{L_i \frac{1}{\beta_l}\mid i=0,\dots,n-1, l=1,\dots,s\right\}\right).$$
Then $L$ annihilates all $\tilde{Q}(\alpha_i)$ and thus has $P$ as a right-hand factor. We summarize the previous discussion as the following algorithm.
\begin{algorithm}
\label{ALG:completelyreducible}
Input: $P\in K\langle \pa_t \rangle$ that is monic and completely reducible.\\[2mm]
Output:  a nonzero $L\in k(t)\langle \pa_t \rangle$ which is divided by $P$ on the right if it exists, otherwise 0.
\begin{enumerate}
\item Write
$$P=\pa_t^n + \sum_{i=0}^{n-1}\frac{a_i}{r}\pa_t^i$$
where $a_i\in k(t)[\bfx], r\in k[\bfx,t]$.
\item Pick $\bfc\in k^m$ such that $r(\bfc)\neq 0$.  By the algorithm in \cite{vanhoeij1}, compute a basis of the solution space $V$ of the equation (\ref{EQ:mixedequation}).
\item Compute a basis of the subspace of $V$ consisting of operators with semisplit coefficients, say $Q_1,\cdots,Q_\ell$.
\item Set $\tilde{Q}=\sum_{i=1}^\ell z_i Q_i$ and using $\tilde{Q}$, compute the matrix $A(\vz)$ as in (\ref{EQ:transformation}).
\item If $\det(A(\vz))=0$ then return 0 and the algorithm terminates. Otherwise compute $\va=(a_1,\dots,a_\ell)\in k^\ell$ such that $\det(A(\va))\neq 0$.
\item Set $b_i$ to be the coefficient of $\pa_t^i$ in $\sum_{j=1}^\ell a_j Q_j$ and write $b_i=\sum_{j=1}^s h_{i,j} \beta_j$ where $h_{i,j}\in k(\bfx)$ and $\beta_j\in k(t)$. Let $L_0=P_{\bfx=\bfc}$ and for each $i=1,\cdots,n-1$ compute $L_i$, the transformation of $L_{i-1}$ by $\pa_t$.
\item Return $\lclm\left(\left\{ L_i \frac{1}{\beta_j} \mid i=0,\dots,n-1, j=1,\dots,s\right\}\right)$.
\end{enumerate}
\end{algorithm}

\subsection{The general case}
Assume that $P$ is $(\bfx,t)$-separable and $P=Q_1Q_2$ where $Q_1,Q_2\in K\langle \pa_t \rangle$. It is clear that $Q_2$ is also $(\bfx,t)$-separable. One may wonder whether $Q_1$ is also $(\bfx,t)$-separable. The following example shows that $Q_1$ may not be $(\bfx,t)$-separable.
\begin{example} Let $K=k(x,t)$ and let $P=\pa_t^2$. Then $P$ is $(\bfx,t)$-separable and
$$\pa_t^2=\left(\pa_t+\frac{x}{xt+1}\right)\left(\pa_t-\frac{x}{xt+1}\right).$$
The operator $\pa_t+x/(xt+1)$ is not $(\bfx,t)$-separable, because $1/(xt+1)$ is one of its solutions and it is not semisplit.
\end{example}
While, the lemma below shows that if $Q_2$ is semisplit then $Q_1$ is also $(\bfx,t)$-separable.
\begin{lem}
\label{LM:composition}
\begin{itemize}
\item [$(1)$]
Assume that $Q_1,Q_2\in K\langle \pa_t \rangle\setminus\{0\}$, and $Q_2$ is semisplit. Then $Q_1Q_2$ is $(\bfx,t)$-separable if and only if both $Q_1$ and $Q_2$ are $(\bfx,t)$-separable.
\item [$(2)$]
  Assume that $P\in K\langle \pa_t\rangle\setminus\{0\}$ and $L$ is a nonzero monic operator in $k(t)\langle \pa_t \rangle$. Then $P$ is $(\bfx,t)$-separable if and only if so is the transformation of $P$ by $L$.
\end{itemize}
\end{lem}
\begin{proof}
Note that the solution space of $\lclm(P_1,P_2)=0$ is spanned by those of $P_1=0$ and $P_2=0$. Hence $\lclm(P_1,P_2)$ is $(\bfx,t)$-separable if and only if so are both $P_1$ and $P_2$.

$(1)$ For the ``only if" part, one only need to prove that $Q_1$ is $(\bfx,t)$-separable.
Assume that $g$ is a solution of $Q_1=0$ in $\calU$. Let $f$ be a solution of $Q_2(y)=g$ in $\calU$. Such $f$ exists because $\calU$ is the universal differential extension of $K$. Then $f$ is a solution of $Q_1Q_2=0$ in $\calU$. By Corollary~\ref{COR:solutions}, $f$ is semisplit. Since $Q_2$ is semisplit, one sees that $g=Q_2(f)$ is semisplit. By Corollary~\ref{COR:solutions} again, $Q_1$ is $(\bfx,t)$-separable.

Now assume that both $Q_1$ and $Q_2$ are $(\bfx,t)$-separable. Let $\tilde{Q}\in K\langle \pa_t \rangle$ be such that $\tilde{Q}Q_2=L$ where $L\in k(t)\langle \pa_t \rangle$ is monic. By Corollary~\ref{cor:semisplitoperators} and the ``only if" part, $\tilde{Q}$ is semisplit and $(\bfx,t)$-separable. Thus $\lclm(Q_1,\tilde{Q})$ is $(\bfx,t)$-separable. Assume that $\lclm(Q_1,\tilde{Q})=N\tilde{Q}$ with $N\in K\langle \pa_t \rangle$. Since $\tilde{Q}$ is semisplit, by the ``only if" part again, $N$ is $(\bfx,t)$-separable. Let $M\in K\langle \pa_t \rangle$ be such that $MN$ is a nonzero operator in $k(t)\langle \pa_t \rangle$. We have that
$$
    M\lclm(Q_1,\tilde{Q})Q_2=MN\tilde{Q}Q_2=MNL\in k(t)\langle \pa_t \rangle.
$$
On the other hand, $M\lclm(Q_1,\tilde{Q})Q_2=M\tilde{M}Q_1Q_2$ for some $\tilde{M}\in K\langle \pa_t \rangle$. Hence $P=Q_1Q_2$ is $(\bfx,t)$-separable.

$(2)$ Since $L$ is $(\bfx,t)$-separable, we have that $P$ is $(\bfx,t)$-separable if and only if so is $\lclm(P,L)$. Let $\tilde{P}$ be the transformation of $P$ by $L$. Then $\tilde{P}L=\lclm(P,L)$. As $L$ is semisplit, the assertion then follows from $(1)$.
\end{proof}

Assume that $P$ is a nonzero operator in $K\langle \pa_t\rangle$. Let $P_0$ be an irreducible right-hand factor of $P$. By Algorithm~\ref{ALG:completelyreducible}, we can decide whether $P_0$ is $(\bfx,t)$-separable or not. Now assume that $P_0$ is $(\bfx,t)$-separable. Then we can compute a nonzero monic operator $L_0\in k(t)\langle \pa_t \rangle$ having $P_0$ as a right-hand factor. Let $P_1$ be the transformation of $P$ by $L_0$. Lemma~\ref{LM:composition} implies that $P$ is $(\bfx,t)$-separable if and only if so is $P_1$. Note that
\begin{align*}
 \ord(P_1)&=\ord(\lclm(P,L_0))-\ord(L_0)\\
 &\leq \ord(P)+\ord(L_0)-\ord(P_0)-\ord(L_0)=\ord(P)-\ord(P_0).
\end{align*}
In other words, $\ord(P_1)<\ord(P)$. Replacing $P$ by $P_1$ and repeating the above process yield an algorithm to decide whether $P$ is $(\bfx,t)$-separable.
\begin{algorithm}
\label{ALG:generalcase}
Input: a nonzeor monic $P\in K\langle \pa_t \rangle$.\\[2mm]
Output:  a nonzero $L\in k(t)\langle \pa_t \rangle$ which is divided by $P$ on the right if it exists, otherwise 0.
\begin{enumerate}
\item If $P=1$ then return 1 and the algorithm terminates.
\item Compute an irreducible right-hand factor $P_0$  of $P$ by algorithms developed in \cite{beke,vanderput-singer,vanhoeij2}.
\item Apply Algorithm~\ref{ALG:completelyreducible} to $P_0$ and let $L_0$ be the output.
\item If $L_0=0$ then return 0 and the algorithm terminates. Otherwise compute the transformation of $P$ by $L_0$, denoted by $P_1$.
\item Apply Algorithm~\ref{ALG:generalcase} to $P_1$ and let $L_1$ be the output.
\item Return $L_1L_0$.
\end{enumerate}
\end{algorithm}
The termination of the algorithm is obvious. Assume that $L_1\neq 0$. Then $L_1=Q_1P_1$ for some $Q_1\in K\langle \pa_t \rangle$. We have that
$P_1L_0=\lclm(P,L_0)$. Therefore $$L_1L_0=Q_1P_1L_0=Q_1\lclm(P,L_0)=Q_1Q_0P$$
for some $Q_0\in K\langle \pa_t \rangle$. This proves the correctness of the algorithm.

\bibliographystyle{abbrv}
\bibliography{parallel-v8}

\begin{thebibliography}{10}

\bibitem{abramov2}
S.~A. Abramov.
\newblock When does {Z}eilberger's algorithm succeed?
\newblock {\em Adv. in Appl. Math.}, 30(3):424--441, 2003.

\bibitem{abramov-le}
S.~A. Abramov and H.~Q. Le.
\newblock A criterion for the applicability of {Z}eilberger's algorithm to
  rational functions.
\newblock {\em Discrete Math.}, 259(1-3):1--17, 2002.

\bibitem{almkvist-zeilberger}
G.~Almkvist and D.~Zeilberger.
\newblock The method of differentiating under the integral sign.
\newblock {\em J. Symbolic Comput.}, 10(6):571--591, 1990.

\bibitem{beke}
E.~Beke.
\newblock Die {I}rreducibilit\"{a}t der homogenen linearen
  {D}ifferentialgleichungen.
\newblock {\em Math. Ann.}, 45(2):278--294, 1894.

\bibitem{bjork}
J.-E. Bj\"{o}rk.
\newblock {\em Rings of differential operators}, volume~21 of {\em
  North-Holland Mathematical Library}.
\newblock North-Holland Publishing Co., Amsterdam-New York, 1979.

\bibitem{bostan-chen-chyzak}
A.~Bostan, S.~Chen, F.~Chyzak, and Z.~Li.
\newblock Complexity of creative telescoping for bivariate rational functions.
\newblock In {\em I{SSAC} 2010---{P}roceedings of the 2010 {I}nternational
  {S}ymposium on {S}ymbolic and {A}lgebraic {C}omputation}, pages 203--210.
  ACM, New York, 2010.

\bibitem{bostan-chen-chyzak-li-xin}
A.~Bostan, S.~Chen, F.~Chyzak, Z.~Li, and G.~Xin.
\newblock Hermite reduction and creative telescoping for hyperexponential
  functions.
\newblock In {\em I{SSAC} 2013---{P}roceedings of the 38th {I}nternational
  {S}ymposium on {S}ymbolic and {A}lgebraic {C}omputation}, pages 77--84. ACM,
  New York, 2013.

\bibitem{bostan-chyzak-lairez-salvy}
A.~Bostan, F.~Chyzak, P.~Lairez, and B.~Salvy.
\newblock Generalized {H}ermite reduction, creative telescoping and definite
  integration of {D}-finite functions.
\newblock In {\em I{SSAC}'18---{P}roceedings of the 2018 {ACM} {I}nternational
  {S}ymposium on {S}ymbolic and {A}lgebraic {C}omputation}, pages 95--102. ACM,
  New York, 2018.

\bibitem{bostan-lairez-salvy}
A.~Bostan, P.~Lairez, and B.~Salvy.
\newblock Creative telescoping for rational functions using the
  {G}riffiths-{D}work method.
\newblock In {\em I{SSAC} 2013---{P}roceedings of the 38th {I}nternational
  {S}ymposium on {S}ymbolic and {A}lgebraic {C}omputation}, pages 93--100. ACM,
  New York, 2013.

\bibitem{chen-chyzak-feng-fu-li}
S.~Chen, F.~Chyzak, R.~Feng, G.~Fu, and Z.~Li.
\newblock On the existence of telescopers for mixed hypergeometric terms.
\newblock {\em J. Symbolic Comput.}, 68(part 1):1--26, 2015.

\bibitem{chen-feng-li-singer}
S.~Chen, R.~Feng, Z.~Li, and M.~F. Singer.
\newblock Parallel telescoping and parameterized {P}icard-{V}essiot theory.
\newblock In {\em I{SSAC} 2014---{P}roceedings of the 39th {I}nternational
  {S}ymposium on {S}ymbolic and {A}lgebraic {C}omputation}, pages 99--106. ACM,
  New York, 2014.

\bibitem{chen-hou-labahn-wang}
S.~Chen, Q.-H. Hou, G.~Labahn, and R.-H. Wang.
\newblock Existence problem of telescopers: beyond the bivariate case.
\newblock In {\em Proceedings of the 2016 {ACM} {I}nternational {S}ymposium on
  {S}ymbolic and {A}lgebraic {C}omputation}, pages 167--174. ACM, New York,
  2016.

\bibitem{chen-kauers}
S.~Chen and M.~Kauers.
\newblock Some open problems related to creative telescoping.
\newblock {\em J. Syst. Sci. Complex.}, 30(1):154--172, 2017.

\bibitem{chen-kauers-koutschan}
S.~Chen, M.~Kauers, and C.~Koutschan.
\newblock Reduction-based creative telescoping for algebraic functions.
\newblock In {\em Proceedings of the 2016 {ACM} {I}nternational {S}ymposium on
  {S}ymbolic and {A}lgebraic {C}omputation}, pages 175--182. ACM, New York,
  2016.

\bibitem{chen-vanhoeij-kauers-koutschan}
S.~Chen, M.~van Hoeij, M.~Kauers, and C.~Koutschan.
\newblock Reduction-based creative telescoping for fuchsian {D}-finite
  functions.
\newblock {\em J. Symbolic Comput.}, 85:108--127, 2018.

\bibitem{chen-hou-mu}
W.~Y.~C. Chen, Q.-H. Hou, and Y.-P. Mu.
\newblock Applicability of the {$q$}-analogue of {Z}eilberger's algorithm.
\newblock {\em J. Symbolic Comput.}, 39(2):155--170, 2005.

\bibitem{euler}
L.~Euler.
\newblock Specimen de constructione aequationum differentialium sine
  indeterminatarum separatione.
\newblock {\em Commentarii academiae scientiarum Petropolitanae}, 6:168--174,
  1733.

\bibitem{kolchin}
E.~R. Kolchin.
\newblock {\em Differential algebra and algebraic groups}.
\newblock Academic Press, New York-London, 1973.
\newblock Pure and Applied Mathematics, Vol. 54.

\bibitem{koutschan}
C.~Koutschan.
\newblock Creative telescoping for holonomic functions.
\newblock In {\em Computer algebra in quantum field theory}, Texts Monogr.
  Symbol. Comput., pages 171--194. Springer, Vienna, 2013.

\bibitem{lairez}
P.~Lairez.
\newblock Computing periods of rational integrals.
\newblock {\em Math. Comp.}, 85(300):1719--1752, 2016.

\bibitem{lang}
S.~Lang.
\newblock {\em Algebra}, volume 211 of {\em Graduate Texts in Mathematics}.
\newblock Springer-Verlag, New York, third edition, 2002.

\bibitem{li-lian-yau}
S.~Li, B.~H. Lian, and S.-T. Yau.
\newblock Picard-{F}uchs equations for relative periods and {A}bel-{J}acobi map
  for {C}alabi-{Y}au hypersurfaces.
\newblock {\em Amer. J. Math.}, 134(5):1345--1384, 2012.

\bibitem{li-wang}
Z.~Li and H.~Wang.
\newblock A criterion for the similarity of length-two elements in a
  noncommutative {PID}.
\newblock {\em J. Syst. Sci. Complex.}, 24(3):580--592, 2011.

\bibitem{lipshitz}
L.~Lipshitz.
\newblock The diagonal of a {$D$}-finite power series is {$D$}-finite.
\newblock {\em J. Algebra}, 113(2):373--378, 1988.

\bibitem{morrison-walcher}
D.~R. Morrison and J.~Walcher.
\newblock D-branes and normal functions.
\newblock {\em Adv. Theor. Math. Phys.}, 13(2):553--598, 2009.

\bibitem{muller-weinzierl-zayadeh}
S.~M\"{u}ller-Stach, S.~Weinzierl, and R.~Zayadeh.
\newblock Picard-{F}uchs equations for {F}eynman integrals.
\newblock {\em Comm. Math. Phys.}, 326(1):237--249, 2014.

\bibitem{ore}
O.~Ore.
\newblock Theory of non-commutative polynomials.
\newblock {\em Ann. of Math. (2)}, 34(3):480--508, 1933.

\bibitem{petkovsek-wilf-zeilberger}
M.~Petkov\v{s}ek, H.~S. Wilf, and D.~Zeilberger.
\newblock {\em {$A=B$}}.
\newblock A K Peters, Ltd., Wellesley, MA, 1996.
\newblock With a foreword by Donald E. Knuth, With a separately available
  computer disk.

\bibitem{singer}
M.~F. Singer.
\newblock Introduction to the {G}alois theory of linear differential equations.
\newblock In {\em Algebraic theory of differential equations}, volume 357 of
  {\em London Math. Soc. Lecture Note Ser.}, pages 1--82. Cambridge Univ.
  Press, Cambridge, 2009.

\bibitem{vanderpoorten}
A.~van~der Poorten.
\newblock A proof that {E}uler missed{$\ldots $}{A}p\'{e}ry's proof of the
  irrationality of {$\zeta (3)$}.
\newblock {\em Math. Intelligencer}, 1(4):195--203, 1978/79.
\newblock An informal report.

\bibitem{vanderput-singer}
M.~van~der Put and M.~F. Singer.
\newblock {\em Galois theory of linear differential equations}, volume 328 of
  {\em Grundlehren der Mathematischen Wissenschaften [Fundamental Principles of
  Mathematical Sciences]}.
\newblock Springer-Verlag, Berlin, 2003.

\bibitem{vanhoeij1}
M.~van Hoeij.
\newblock Rational solutions of the mixed differential equation and its
  application to factorization of differential operators.
\newblock In E.~Engeler, B.~F. Caviness, and Y.~N. Lakshman, editors, {\em
  Proceedings of the 1996 International Symposium on Symbolic and Algebraic
  Computation, {ISSAC} '96, Zurich, Switzerland, July 24-26, 1996}, pages
  219--225. {ACM}, 1996.

\bibitem{vanhoeij2}
M.~van Hoeij.
\newblock Factorization of differential operators with rational functions
  coefficients.
\newblock {\em J. Symbolic Comput.}, 24(5):537--561, 1997.

\bibitem{weinstraub}
S.~H. Weintraub.
\newblock {\em Differential forms:Theory and practice}.
\newblock Elsevier/Academic Press, Amsterdam, second edition, 2014.
\newblock Theory and practice.

\bibitem{wilf-zeilberger1}
H.~S. Wilf and D.~Zeilberger.
\newblock Rational functions certify combinatorial identities.
\newblock {\em J. Amer. Math. Soc.}, 3(1):147--158, 1990.

\bibitem{wilf-zeilberger2}
H.~S. Wilf and D.~Zeilberger.
\newblock An algorithmic proof theory for hypergeometric (ordinary and
  ``{$q$}'') multisum/integral identities.
\newblock {\em Invent. Math.}, 108(3):575--633, 1992.

\bibitem{wilf-zeilberger3}
H.~S. Wilf and D.~Zeilberger.
\newblock Rational function certification of multisum/integral/``{$q$}''
  identities.
\newblock {\em Bull. Amer. Math. Soc. (N.S.)}, 27(1):148--153, 1992.

\bibitem{zeilberger}
D.~Zeilberger.
\newblock The method of creative telescoping.
\newblock {\em J. Symbolic Comput.}, 11(3):195--204, 1991.

\end{thebibliography}
\end{document}